\documentclass{article}

\usepackage{authblk}
\usepackage{amsmath}
\usepackage{amssymb}
\usepackage{amsthm}

\theoremstyle{definition}
\newtheorem*{defn}{Definition}

\theoremstyle{plain}
\newtheorem{lem}{Lemma}
\newtheorem{thm}{Theorem}

\DeclareMathOperator{\spn}{span}
\DeclareMathOperator{\rk}{rk}
\DeclareMathOperator{\id}{id}
\DeclareMathOperator{\im}{im}
\DeclareMathOperator{\CSS}{CSS}
\DeclareMathOperator{\GL}{GL}

\title{Qudit homological product codes}

\author[1]{M\'at\'e Farkas\thanks{\texttt{mate.farkas@phdstud.ug.edu.pl}}}
\author[2]{P\'eter Vrana\thanks{\texttt{vranap@math.bme.hu}}}
\affil[1]{Department of Theoretical Physics, Budapest University of Technology and Economics, Budafoki \'ut 8., 1111 Budapest, Hungary}
\affil[2]{Department of Geometry, Budapest University of Technology and Economics, Egry J\'ozsef u. 1., 1111 Budapest, Hungary}

\begin{document}
\maketitle

\begin{abstract}
In this note we show that the random homological product code construction of Bravyi and Hastings can be extended to qudits of dimension $D$ with $D$ an odd prime. While the result is not surprising, the proof does require new ideas.
\end{abstract}

\section{Introduction}

Recently it has been shown that the homological product of two random chain 
complexes gives rise to good CSS codes with stabilizer weight $O(\sqrt{n})$, 
where $n$ is the code size. \cite{BravyiHastingsACM,BravyiHastings} The proof
makes use of a simplified version of the homological product, called the
``single sector theory''. In this construction a chain complex is a pair
$(C,\partial)$ with $C$ a free module over a ring $R$ with a distinguished
basis and the linear map $\partial:C\to C$ satisfies $\partial^2=0$. When $2=0$
in $R$ (that is, if $R=\mathbb{Z}_2$, we have that $1+1=0$), it is possible to
define tensor products of such chain complexes as $(C_1\otimes C_2,\partial_1\otimes I_2+I_1\otimes\partial_2)$. It is then shown that if $R=\mathbb{Z}_2$ and $\partial_1$ and $\partial_2$ are random conjugates of some fixed boundary operator, then the product complex leads to good codes with high probability.

Unfortunately, when $2\neq 0$, $\partial_1\otimes I_2+I_1\otimes\partial_2$ is not a boundary operator in general, since its square is $2\partial_1\otimes\partial_2$. In typical applications of homological algebra, the modules $C_1$ and $C_2$ are graded and boundary maps are homogeneous of degree $-1$ (or $+1$), and so it is possible to remedy the situation by introducing signs depending on the grading. In ref. \cite{BravyiHastings} this variant is called the ``multiple sector theory''. However, in this case it seems to be difficult to find good lower bounds on the distance of the homological product code. Our main result is the analysis of an intermediate structure, one which is rich enough to make it possible to form tensor products even in the qudit case, and at the same time it is sufficiently close to the single sector variant so that the proofs can be modified to reach essentially the same conclusion.

The structure of this paper is as follows. In section \ref{sec:double} we describe the above mentioned intermediate structure, and in section \ref{sec:example}
we give an example of the construction. In section \ref{sec:random} we extend the proof of Bravyi and Hastings to qudit codes. 
Here we mainly focus on proofs that are significantly different from those of
the single sector theory. While the main line of argument is presented here,
the reader may refer to \cite{BravyiHastings} for further details.

\section{Double sector theory}\label{sec:double}

To define the ``single sector'' theory, one starts with a single $\mathbb{Z}_D$-module $C$ with a basis and a boundary operator $\delta:C\to C$, which satisfies $\delta^2=0$. In this case the columns and rows of (the matrix of) $\delta$ give $Z$ and $X$-type stabilizer generators.

When $D=2$, the tensor product of two chain complexes $(C_i,\delta_i)$ can be defined in the single sector theory simply as $(C,\partial)$ with $C=C_1\otimes C_2$ and $\partial=\delta_1\otimes I_2+I_1\otimes\delta_2$. This actually works over any ring as long as $2=0$. Ref. \cite{BravyiHastingsACM} uses the single sector theory to form the product of two independent random chain complexes, because in this framework the distance of the resulting CSS code is easier to analyze.

When $D>2$, it is not possible to form a tensor product of chain complexes in the single sector variant. However, it is possible to introduce some extra structure which lies between the single and multiple sector theories in the sense that it allows us to form tensor products, but it is not much more difficult to find lower bounds on the distance of the product of suitably defined random codes. One possibility is given by the following definition
(see e.g. \cite{HiltonStammbach}):
\begin{defn}
A differential $\mathbb{Z}_D$-module with involution (or chain complex with involution over $\mathbb{Z}_D$) is a triple $(C,\partial,P)$ where $C$ is a module over $\mathbb{Z}_D$, $\partial:C\to C$ and $P:C\to C$ are $\mathbb{Z}_D$-linear and satisfy
\begin{equation}\label{eq:complexwithinvolution}
\begin{split}
\partial^2 & = 0  \\
P^2 & = I  \\
\partial P+P\partial & = 0.
\end{split}
\end{equation}

The tensor product of two differential $\mathbb{Z}_D$-modules with involution $(C_1,\delta_1,P_1)$ and $(C_2,\delta_2,P_2)$ is defined to be $(C,\partial,P)$ where $C=C_1\otimes C_2$, $\partial=\delta_1\otimes I_2+P_1\otimes\delta_2$ and $P=P_1\otimes P_2$.
\end{defn}

With this definition, the tensor product also satisfies eqn. \eqref{eq:complexwithinvolution}. When $D=2$, the single sector theory suffices, whereas when $D$ is not a prime, additional difficulties arise, and we do not know how these can be overcome. For this reason, from now on we will assume that $D$ is an odd prime. In this case $\mathbb{Z}_D$ is a field, therefore the appearing modules are in fact vector spaces. In addition, $2$ is invertible, so we can form the linear combinations
\begin{equation}
P_+=\frac{I+P}{2}\text{ and }P_-=\frac{I-P}{2},
\end{equation}
where $I$ is the identity operator. These satisfy $P_+^2=P_+$, $P_-^2=P_-$ and $P_+P_-=P_-P_+=0$. Moreover, $P_++P_-=I$ and $P_+-P_-=P$. This implies that $C$ can be written as a direct sum $C=C_+\oplus C_-$ with $C_\pm=P_\pm C$, these are invariant subspaces of $P$ and $P|_{C_\pm}=\pm \id_{C_\pm}$. It also follows that $\partial P_\pm=P_\mp\partial$, therefore $\partial$ takes $C_\pm$ to $C_\mp$. Equivalently, we could have started with a pair of spaces $C_+$ and $C_-$ and a pair of linear maps $\partial_{+-}:C_-\to C_+$ and $\partial_{-+}:C_+\to C_-$ satisfying $\partial_{+-}\partial_{-+}=\partial_{-+}\partial_{+-}=0$. In line with the naming convention of ref. \cite{BravyiHastings}, this may be called the ``double sector theory''.

The key feature of the homological product is that it does not increase the stabilizer weights too much. Let $w(A)$ denote the maximal weight over the rows and columns of a matrix $A$ (when $A$ is a linear map, this depends on the basis, but we will always use a fixed basis). Clearly, $w(A+B)\le w(A)+w(B)$ and $w(A\otimes B)=w(A)w(B)$ holds with respect to the product basis. In the single sector theory these imply that $w(\partial)\le w(\delta_1)+w(\delta_2)$, but in our case this relation becomes
\begin{equation}
w(\partial)\le w(\delta_1)+w(P_1)w(\delta_2).
\end{equation}
In order to get the best possible bound, we will choose $P_i$ to be (in block matrix form)
\begin{equation}
P_i=\begin{bmatrix}
I & 0 \\
0 & -I
\end{bmatrix}.
\end{equation}
Then $w(P_i)=1$, therefore the same inequality holds as in the single sector theory. An additional advantage of this choice is that $C_+$ and $C_-$ are now spanned by subsets of the basis vectors. In the next section we would like to consider multiple boundary operators on a single space. One possibility is to choose a fixed boundary $\delta$ and conjugate it with arbitrary invertible matrices. In order to retain compatibility with the involution, we need to conjugate that with the same matrix as well, but this would potentially increase its row and column weights. This can be prevented by requiring that the invertible matrix commutes with $P$. In a block matrix form, the commutator reads
\begin{equation}
\begin{bmatrix}
I & 0 \\
0 & -I
\end{bmatrix}\cdot\begin{bmatrix}
A & B \\
C & D
\end{bmatrix}
-\begin{bmatrix}
A & B \\
C & D
\end{bmatrix}\cdot\begin{bmatrix}
I & 0 \\
0 & -I
\end{bmatrix}=
\begin{bmatrix}
0 & 2B \\
-2C & 0
\end{bmatrix}.
\end{equation}
Since $2$ is invertible, the commutator is $0$ iff $B=0$ and $C=0$. The conjugating matrix is therefore block diagonal. Let us call its diagonal blocks $U_+$ and $U_-$. Then the boundary operator transforms as
\begin{equation}
\begin{bmatrix}
U_+ & 0 \\
0 & U_-
\end{bmatrix}\cdot
\begin{bmatrix}
0 & \delta_{+-} \\
\delta_{-+} & 0
\end{bmatrix}\cdot
\begin{bmatrix}
U_+^{-1} & 0 \\
0 & U_-^{-1}
\end{bmatrix}
=
\begin{bmatrix}
0 & U_+\delta_{+-}U_-^{-1} \\
U_-\delta_{-+}U_+^{-1} & 0
\end{bmatrix}.
\end{equation}
It is possible to show that both $\delta_{+-}$ and $\delta_{-+}$ can be brought to a standard form simultaneously, which only depend on the two homological dimensions. We will choose random boundary operators by fixing a standard $\delta_0$ and letting $\delta_{+-}=U_+\delta_0 U_-^{-1}$ and $\delta_{-+}=U_-\delta_0 U_+^{-1}$ with $U_+$ and $U_-$ drawn uniformly from the set of invertible matrices. We choose $\delta_0$ to be
\begin{equation}\label{eq:standardboundary}
\delta_0=\begin{bmatrix}
0 & 0 & 0  \\
0 & 0 & I  \\
0 & 0 & 0
\end{bmatrix}
\end{equation}
in block matrix form with row/column sizes $H$, $L$, $L$ in this order, where $H$ is the homological dimension and $H+2L=\dim C_+=\dim C_-$. In particular, we choose these parameters to be the same for the $+$ and $-$ blocks, but it would be possible to allow different values.

Once we have two chain complexes with involution and form the tensor product, it would be possible to forget the involution and construct a CSS code as in the single sector theory. Unfortunately, the code distance of the product of two random complexes seems to be difficult to bound in this case. Instead of this, we adopt a definition which resembles the multiple sector theory. Physical qudits will correspond to the basis elements of $C_+$, the columns of $\partial_{+-}$ will be $Z$-type generators and the rows of $\partial_{-+}$ will be the $X$-type generators (for this we need the special form of $P$ given above). The code obtained this way will be denoted by $\CSS(C,\partial,P)$.

\section{Example}\label{sec:example}

In this section, we give an example of the construction, and calculate the code
properties numerically. For the sake of simplicity, we begin with the most
compact code that fits in our framework: a $[3,1,2,3]$ qutrit code with
stabilizer generators $XXX$ and $ZZZ$. This code has distance 2, so it corrects
one erasure error, if the position is known \cite{Muralidharan}. In the
following, we construct the homological product of this code with itself. Note
that $\delta$ defines this code, and is a valid boundary operator, if its blocks
are given by
\begin{equation}\label{eq:312deltapm}
\delta_{+-}=\delta_{-+}=
\begin{bmatrix}
1 & 1 & 1 \\
1 & 1 & 1 \\
1 & 1 & 1
\end{bmatrix},
\end{equation}
thus the boundary operator is
\begin{equation}\label{eq:312delta}
\delta=
\begin{bmatrix}
0 & 0 & 0 & 1 & 1 & 1 \\
0 & 0 & 0 & 1 & 1 & 1 \\
0 & 0 & 0 & 1 & 1 & 1 \\
1 & 1 & 1 & 0 & 0 & 0 \\
1 & 1 & 1 & 0 & 0 & 0 \\
1 & 1 & 1 & 0 & 0 & 0
\end{bmatrix}.
\end{equation}

Now, let us construct the boundary operator of the product code, i.e. $\partial
=\delta\otimes I +P_1\otimes\delta$, where
\begin{equation}\label{eq:P3}
P_1=
\begin{bmatrix}
I & 0 \\
0 & -I
\end{bmatrix},
\end{equation}
thus in block form
\begin{equation}\label{eq:1814delta}
\partial=
\begin{bmatrix}
\delta & 0 & 0 & I & I & I \\
0 & \delta & 0 & I & I & I \\
0 & 0 & \delta & I & I & I \\
I & I & I & -\delta & 0 & 0 \\
I & I & I & 0 & -\delta & 0 \\
I & I & I & 0 & 0 & -\delta
\end{bmatrix}.
\end{equation}
From here, it is immediate that the weight of the product code is $w=6$. We
obtained the distance by an exhaustive search over all non-trivial cycles and
cocycles of $\partial$, and it turns out that the distance of this product code
is $d=4$. Thus, we obtained a $[18,1,4,6]$ qutrit code. Note that the stabilizer
weight is indeed low. As a comparison, concatenating the code with itself
would result in full ($w=n$) stabilizer weight for any level of concatenation.

In order to see the stabilizers, we need to bring our code to the standard
form described in section \ref{sec:double}. That is, we need to perform a 
permutation $\Pi$, such that
\begin{equation}\label{eq:transf}
\Pi (P_1\otimes P_2) \Pi^\dagger =
\begin{bmatrix}
I & 0 \\
0 & -I
\end{bmatrix}.
\end{equation}
Once we solve this equation, we get a boundary operator
\begin{equation}\label{eq:transfdelta}
\Pi \partial \Pi^\dagger =
\begin{bmatrix}
0 & \partial_{-+} \\
\partial_{+-} & 0
\end{bmatrix}.
\end{equation}
The columns of $\partial_{-+}$ will give us the $Z$-type generators, while the
rows of $\partial_{+-}$ the $X$-type generators. For example, the first column
of $\partial_{-+}$ and the last row of $\partial_{+-}$ give, respectively:
\begin{align}\label{eq:prodstabilizer}
&ZZZIIIIIIZIIZIIZII, \\
II&XIIXIIXIIIIIIX^2X^2X^2.
\end{align}

\section{Homological product of two random chain complexes with involutions}\label{sec:random}

In this section we let $n=H+2L$ with $H=\lfloor\rho n\rfloor$ ($\rho$ is
defined in Lemma \ref{lem:randomcode},
but it will turn out, that it is -- up to the floor function --
the encoding rate, i.e. the ratio of logical and physical qudits in our code),
and $\delta_0$ is the matrix from eq.
\eqref{eq:standardboundary}. By a random $2n\times 2n$ boundary operator on
$C=C_+\oplus C_-$ with $C_\pm=\mathbb{Z}_D^n$ we mean
\begin{equation}\label{eq:randomconjugate}
\delta=\begin{bmatrix}
U_+ & 0  \\
0 & U_-
\end{bmatrix}
\begin{bmatrix}
0 & \delta_0  \\
\delta_0 & 0
\end{bmatrix}
\begin{bmatrix}
U_+^{-1} & 0  \\
0 & U_-^{-1}
\end{bmatrix}
\end{equation}
in block form. The boundary operator in the middle will be called the standard one.

\begin{lem}\label{lem:randomcode}
For any $\varepsilon>0$ there exist $c,\rho>0$ such that the probability that the kernel of a random $2n\times 2n$ boundary operator with $H\le\rho n$ contains a nonzero vector with weight less than $cn$ is $O(D^{(-1/2+\varepsilon)n})$.
\end{lem}
\begin{proof}
Let $\delta$ be as in eq. \eqref{eq:randomconjugate}. Since
\begin{equation}
U=\begin{bmatrix}
U_+ & 0  \\
0 & U_-
\end{bmatrix}
\end{equation}
is invertible, $\ker\delta=U\ker(U^{-1}\delta U)$. $P$ anticommutes with $\delta$, therefore if $v\in\ker\delta$ then $v_\pm=\frac{1}{2}(I\pm P)v\in\ker\delta$, which in turn holds iff $U^{-1}v_\pm\in\ker(U^{-1}\delta U)$. 

For any nonzero vector $w\in \ker(U^{-1}\delta U)$ there are three possibilities according to whether the vectors $w_\pm=\frac{1}{2}(I\pm P)w$ vanish or not. If e.g. $w_+\neq0$ and $w_-=0$, then $Uw$ is distributed uniformly on the set of vectors $v$ with $v_+\neq 0$ and $v_-=0$. When $c<1-\frac{1}{D}$, the probability that such a vector has weight less than $cn$ is at upper bounded by
\begin{equation}
\frac{1}{D^n-1}\sum_{1\le w<cn}\binom{n}{w}(D-1)^w\le O(1)\cdot D^{(-(1-c)\log_D(1-c)-c\log_Dc+c\log_D(D-1)-1)n},
\end{equation}
using Stirling's formula to get the inequality. We get the same probability when $w_+=0$ and $w_-\neq 0$. When $w_+\neq 0\neq w_-$, the vector $Uw$ is uniform on vectors $v$ with $v_\pm\neq 0$. This time the probability of having weight less than $cn$ is upper bounded by
\begin{equation}
O(1)\cdot D^{(-(1-c)\log_D(1-c)-c\log_Dc+c\log_D(D-1)-1)2n}.
\end{equation}

Finally, we apply the union bound. The number of vectors of the first two types is $O(1)\cdot D^{\frac{1}{2}n+\frac{1}{2}\lfloor\rho n\rfloor}$ and the number of vectors of the third type is $O(1)\cdot D^{n+\lfloor\rho n\rfloor}$, therefore the probability of having a nonzero vector with weight less than $cn$ is at most
\begin{equation}
O(1)\cdot D^{(-(1-c)\log_D(1-c)-c\log_Dc+c\log_D(D-1)-\frac{1}{2}+\frac{\rho}{2})n}.
\end{equation}
For any $D>0,\epsilon>0$ and small enough $c,\rho$ the exponent is less than $(-\frac{1}{2}+\epsilon)n$.
\end{proof}

\begin{lem}\label{lem:reducedcomplex}
Let $(C,\delta,P)$ be a chain complex with involution, $\mathcal{V}\le C$ a subspace and $\mathcal{V}^>$ a direct complement. Let $W,W^>:C\to C$ be the projections corresponding to this direct sum decomposition and suppose that $PW=WP$. Let $S^>=W\delta(\mathcal{V}^>)$ and $\mathcal{V}'=\mathcal{V}/S^>$. Let $\varphi:C\to\mathcal{V}'$ be defined as $h\mapsto Wh+S^>$. Then the maps $\delta',P':\mathcal{V}'\to\mathcal{V}'$ given by
\begin{equation}
\delta'(x+S^>)=\varphi(\delta(x))\text{ and }P'(x+S^>)=\varphi(P(x))
\end{equation}
are well defined, $(\mathcal{V}',\delta',P')$ is a chain complex with involution, and $\varphi:C\to\mathcal{V}'$ is a chain map satisfying $\varphi P=P'\varphi$. Moreover, the equalities
\begin{equation}\label{eq:reducedkerim}
\ker\delta'=\varphi(\delta^{-1}(\mathcal{V}^>))\text{ and }\im\delta'=\varphi(\im\delta)
\end{equation}
hold.
\end{lem}
\begin{proof}
Suppose that $x\in S^>$, i.e. there is some $y\in\mathcal{V}^>$ such that $x=W\delta y$. Then $W\delta x=W\delta W\delta y=W\delta(W-I)\delta y=-W\delta W^>\delta y\in S^>$, therefore $\varphi(\delta x)=W\delta x+S^>=S^>$. Similarly, $\varphi(P x)=\varphi(PW\delta y)=\varphi(-W\delta Py)\in S^>$, because $Py=PW^>y=W^>Py\in\mathcal{V}^>$.

By definition, $\delta'\varphi=\varphi\delta$ and $P'\varphi=\varphi P$. This implies that $\delta'\delta'\varphi=\varphi\delta\delta=0$, $P' P'\varphi=\varphi P P=\varphi$ and $(P'\delta'+\delta'P')\varphi=\varphi(P\delta+\delta P)=0$. By surjectivitiy of $\varphi$, the equalities ${\delta'}^2=0$, ${P'}^2=I$ and $P'\delta'+\delta'P'=0$ follow.

The proof of eq. \eqref{eq:reducedkerim} is the same as in \cite[Lemma 8]{BravyiHastings} (Lemma 10 in \cite{BravyiHastingsACM}).
\end{proof}

\begin{defn}
A boundary operator $\delta:C_+\oplus C_-\to C_+\oplus C_-$ is called good if neither $\ker\delta\cap C_+$ nor $\ker\delta\cap C_-$ contains a nonzero vector supported on the last $n-n'$ coordinates.
\end{defn}

In the following $\mathcal{V}$ will be the subspace of $C$ spanned by the first $n'$ basis vectors in $C_\pm$ and $\mathcal{V}^>$ the direct complement spanned by the remainig ones. Then the conditions of lemma \ref{lem:reducedcomplex} are satisfied.

\begin{lem}
Let $\delta$ be a good boundary operator. Then $\dim\mathcal{V}'_{\pm}=2n'-n$ and
\begin{align}
\dim(\ker\delta'_{-+}) & =\dim(\ker\delta_{-+})-(n-n')  \\
\dim(\im\delta'_{-+}) & =\dim(\im\delta_{-+})-(n-n')
\end{align}
hold, and similarly for $\delta'_{+-}$.
\end{lem}
\begin{proof}
The proof is very similar to that of \cite[Lemma 9]{BravyiHastings} (Lemma 11 in \cite{BravyiHastingsACM}).
\end{proof}

Let $C_1=C_2=\mathbb{Z}_D^{2n}$ equipped with $P_1=P_2$ which is diagonal in the standard basis, and acts as $+1$ ($-1$) on the first (last) $n$ coordinates, and let $\delta_1,\delta_2$ be compatible boundary operators. Suppose that $\psi\in C_+\cap\ker{\partial}$ where $C=C_1\otimes C_2$, $C_+=C_{1+}\otimes C_{2+}\oplus C_{1-}\otimes C_{2-}\le C$ and $\partial=\delta_1\otimes I_1+P_1\otimes \delta_2$. $\psi$ can be thought of as a block-diagonal matrix with two $n\times n$ blocks $\psi_+$ and $\psi_-$.

If the weight of $\psi$ is less than $cn^2$ then the same is true for both blocks. Choosing some $r$ such that $c<r<1$, we can find at least $n'=(1-r)n$ rows and columns in $\psi_\pm$ having weight at most $cnr^{-1}$. The two $n'\times n'$ matrices obtained this way will be called the reduced matrix of $\psi$. An $n'\times n'$ matrix having row and column weights at most $c'n'$ (with $c'=cr^{-1}/(1-r)$) is said to satisfy the uniform low weight condition.

In the following lemma, $\mathcal{V}_i$ and $\mathcal{V}^>_i$ are the subspaces of $C_i$ defined similarly as before.
\begin{lem}
Suppose that the distance of $\CSS(C_i,\delta_i,\pm P_i)$ is at least $2(n-n')+1$ ($i=1,2$). If $h\in C_+\cap\ker\partial$ has vanishing reduced matrix then $h\in\im\partial$. A similar statement holds for cocycles.
\end{lem}
\begin{proof}
The proof goes along the lines of \cite[Lemma 5]{BravyiHastings} (Lemma 5 of \cite{BravyiHastingsACM}), but there are differences due to the fact that we restrict to the subspace $C_+$. Let $\bar{h}_{a+}\in\ker\delta_{a-+}^T\setminus\im\delta_{a+-}^T$ and $\bar{h}_{a-}\in\ker\delta_{a-+}^T\setminus\im\delta_{a+-}^T$ be nontrivial cocycles ($a=1,2$) and let $S=\{n'+1,n'+2,\ldots,n\}\cup\{n+n'+1,n+n'+2,\ldots,2n\}$. Since $|S|=2(n-n')<2(n-n')+1$, the cleaning lemma (\cite[Lemma 1]{BravyiTerhal}, the same proof works for qudits) implies that there exist $\bar{\omega}_{a+}\in\im\delta_{a+-}^T$ and $\bar{\omega}_{a-}\in\im\delta_{a-+}^T$ such that the support of $\bar{h}_{a\pm}+\bar{\omega}_{a\pm}$ is disjoint from $S$.

Choose nontrivial cocycles supported outside $S$ such that their cosets span the cohomology groups:
\begin{equation}
\begin{split}
\ker\delta^T_{a-+} & =\spn(\bar{h}^1_{a+},\bar{h}^2_{a+},\ldots,\bar{h}^H_{a+})+\im\delta^T_{a+-}  \\
\ker\delta^T_{a+-} & =\spn(\bar{h}^1_{a-},\bar{h}^2_{a-},\ldots,\bar{h}^H_{a-})+\im\delta^T_{a-+}.
\end{split}
\end{equation}
Take the dual basis of nontrivial cycles:
\begin{equation}
\begin{split}
\ker\delta_{a-+} & =\spn(h^1_{a+},h^2_{a+},\ldots,h^H_{a+})+\im\delta_{a+-}  \\
\ker\delta_{a+-} & =\spn(h^1_{a-},h^2_{a-},\ldots,h^H_{a-})+\im\delta_{a-+},
\end{split}
\end{equation}
such that $(\bar{h}^i_{a\pm},h^j_{a\pm})=\delta_{ij}$. By the K\"unneth formula we have (see e.g. \cite{HiltonStammbach})
\begin{equation}
\begin{split}
\ker\partial\cap C_+=
 & \spn\{h^i_{1+}\otimes h^j_{2+}|1\le i,j\le H\}  \\
 & + \spn\{h^i_{1-}\otimes h^j_{2-}|1\le i,j\le H\}+\im\partial\cap C_+  \\
\ker\partial^T\cap C_+=
 & \spn\{\bar{h}^i_{1+}\otimes \bar{h}^j_{2+}|1\le i,j\le H\}  \\
 & + \spn\{\bar{h}^i_{1-}\otimes \bar{h}^j_{2-}|1\le i,j\le H\}+\im\partial^T\cap C_+.
\end{split}
\end{equation}

Let $h\in\ker\partial\cap C_+$ be a cycle with vanishing reduced matrix. Then
\begin{equation}
h=\sum_{i,j=1}^H x^+_{i,j}h^i_{1+}\otimes h^j_{2+}+\sum_{i,j=1}^H x^-_{i,j}h^i_{1-}\otimes h^j_{2-}+\omega,
\end{equation}
where $\omega\in\im\partial\cap C_+$ and $x^\pm_{i,j}=(\bar{h}^i_{1\pm}\otimes\bar{h}^j_{2\pm},h)\in\mathbb{Z}_D$. Since $\bar{h}^i_{1\pm}\otimes\bar{h}^j_{2\pm}$ is supported on the reduced matrix where $h$ vanishes, $x^\pm_{i,j}=0$, therefore $h=\omega\in\im\partial$.

The proof for cocycles is the same with the roles of $\delta$ and $\delta^T$ reversed.
\end{proof}

\begin{defn}
We let $E_{a,b,r}^{A,B,R}$ denote the number of rank $R$ matrices of size $A\times B$ which extend an (arbitrary) rank $r$ matrix of size $a\times b$ over $\mathbb{Z}_D$.
\end{defn}

\begin{lem}\label{lem:Ebound}
\begin{equation}
E^{A,B,R}:=E_{0,0,0}^{A,B,R}=\Theta(1)\cdot D^{(A+B)R-R^2},
\end{equation}
and
\begin{equation}
E_{a,b,r}^{A,B,R}=O(1)\cdot D^{(A+B-b)R-br-R^2+(b-a+r+R)^2/4}.
\end{equation}
\end{lem}
\begin{proof}
The proof is very similar to that in \cite[Appendix A]{BravyiHastings} (Proposition 1 in \cite{BravyiHastingsACM}).
\end{proof}

\begin{lem}\label{lem:Zbound}
Let $\delta_1,\delta_2$ be boundary operators with $\dim\im\delta_a=2L$ and $\dim\ker\delta_a=2(L+H)$ and let $\partial=\delta_1\otimes I_2+P_1\otimes\delta_2$. Let $Z_{H,L}(r_+,r_-)$ be the number of $h\in\ker\partial\cap C_+$ such that $\rk h_\pm=r_\pm$. Then $Z_{H,L}(r_+,r_-)$ only depends on $r_+,r_-,H$ and $L$ and
\begin{equation}
Z_{H,L}(r_+,r_-)\le O(n)\cdot D^{2(H+L)(r_++r_-)-(r_+^2+r_-^2)}\sum_{l=0}^{2L} D^{-l^2+(r_++r_--2H)l}.
\end{equation}
\end{lem}
\begin{proof}
Since $\delta_1$ and $\delta_2$ are conjugates of the standard boundary operator $\delta$ and left and right multiplication by a block diagonal matrices $U_1,U_2$ does not change the ranks, we may assume that $\delta_a$ are the standard ones. In this case $C_+\cap\ker\partial$ is equal to the set of matrices of the following form:
\begin{equation}
h=\begin{bmatrix}
h_+ & 0  \\
0 & h_-
\end{bmatrix}
\text{ where }
h_+=\begin{bmatrix}
A_+ & B_+ & 0  \\
C_+ & D_+ & F  \\
0 & G & 0
\end{bmatrix}
\text{, }
h_-=\begin{bmatrix}
A_- & B_- & 0  \\
C_- & D_- & G  \\
0 & -F & 0
\end{bmatrix}.
\end{equation}

To count the number of such matrices with given pair of ranks, let us first fix $F$ and $G$ with $f=\rk F$ and $g=\rk G$. These ranks must satisfy $0\le f+g\le\min\{r_+,r_-\}$ and $f,g\le L$, but are otherwise completely arbitrary. Let $U,V,X,Y$ be invertible $L\times L$ matrices such that $UFY$ and $VGX$ have nonzero elements only in their upper left $f\times f$ and $g\times g$ corners, respectively. Apply the transformations
\begin{equation}
h_+\mapsto\begin{bmatrix}
I & 0 & 0  \\
0 & U & 0  \\
0 & 0 & V
\end{bmatrix}
h_+\begin{bmatrix}
I & 0 & 0  \\
0 & X & 0  \\
0 & 0 & Y
\end{bmatrix}
\text{, }
h_-\mapsto\begin{bmatrix}
I & 0 & 0  \\
0 & V & 0  \\
0 & 0 & U
\end{bmatrix}
h_-\begin{bmatrix}
I & 0 & 0  \\
0 & Y & 0  \\
0 & 0 & X
\end{bmatrix}.
\end{equation}
This does not change the block structure or the ranks. Removing the rows and columns from $h_+$ and $h_-$ corresponding to the nonempty rows and columns of the new $F$ and $G$ decreases both ranks by $f+g$ independently of $A_\pm,B_\pm,C_\pm,D_\pm$. The remaining nonzero matrices have sizes $(H+L-f)\times(H+L-g)$ and $(H+L-g)\times(H+L-f)$. These can be chosen arbitrarily as long as their ranks are $r_\pm-f-g$. Summing over the possible choices gives
\begin{equation}
\begin{split}
Z_{H,L}(r_+,r_-)
 = &\sum_{\substack{f,g=0  \\ f+g\le\min\{r_+,r_-\}}}^L E^{L,L,f}E^{L,L,g}D^{2(f+g)(H+L)-2fg}  \\
 & \cdot E^{H+L-f,H+L-g,r_+-f-g}E^{H+L-g,H+L-f,r_--f-g}  \\
 = &O(1)\cdot D^{2(H+L)(r_++r_-)-r_+^2-r_-^2}  \\
 & \cdot\sum_{\substack{f,g=0  \\ f+g\le\min\{r_+,r_-\}}}^L D^{-(f+g)^2+(f+g)(r_++r_--2H)}.
\end{split}
\end{equation}
Introducing $l=f+g$, we may replace the sum over $f$ and $g$ with $n$ times a sum over $l$ with $0\le l\le\min\{2L,r_+,r_-\}$ to get an upper bound.
\end{proof}

\begin{lem}\label{lem:gammasum}
Let $\delta_1,\delta_2$ be good boundary operators with $C_{i\pm}=\mathbb{Z}_D^n$ and homological dimension $H+H$. Let $\Gamma(R_+,R_-)$ be the number of reduced cycles $(h_+,h_-)$ with ranks $\rk h_+=R_+$ and $\rk h_-=R_-$. Then
\begin{equation}
\Gamma(R_+,R_-)=\sum_{r_+=0}^{\min\{K,R_+\}}\sum_{r_-=0}^{\min\{K,R_-\}}Z_{H,L-(n-n')}(r_+,r_-)E_{K,K,r_+}^{n',n',R_+}E_{K,K,r_-}^{n',n',R_-},
\end{equation}
where $K=2n'-n$.
\end{lem}
\begin{proof}
From lemma \ref{lem:reducedcomplex} we get a map $\varphi\otimes\varphi:C\to\mathcal{V}'_1\otimes\mathcal{V}'_2$ with $C=C_1\otimes C_2$ and $\mathcal{V}'_i$ as in the lemma. The proof of \cite[Lemma 10]{BravyiHastings} works in our case without modification, and it implies (after restricting to the subspaces $C_+=C_{1+}\otimes C_{2+}\oplus C_{1-}\otimes C_{2-}$ and $\mathcal{V}'_+=\mathcal{V}'_{1+}\otimes \mathcal{V}'_{2+}\oplus \mathcal{V}'_{1-}\otimes \mathcal{V}'_{2-}$) that
\begin{equation}
\Gamma(R_+,R_-)=\sum_{h\in \mathcal{V}'_+\cap\ker\partial'}|\{g_+\oplus g_-\in\mathcal{V}_+|\rk g_\pm=R_\pm\text{ and }(\varphi\otimes\varphi)g=h\}|.
\end{equation}
A similar argument as in the qubit case allows us to rewrite the above expression as
\begin{equation}
\begin{split}
\Gamma(R_+,R_-)
 & =\sum_{r_\pm=0}^{\min\{K,R_\pm\}}|\{h\in\mathcal{V}'_+\cap\ker\partial'|\rk h_\pm=r_\pm\}|\cdot E_{K,K,r_+}^{n',n',R_+}E_{K,K,r_-}^{n',n',R_-}  \\
 & =\sum_{r_+=0}^{\min\{K,R_+\}}\sum_{r_-=0}^{\min\{K,R_-\}}Z_{H,L-(n-n')}(r_+,r_-)E_{K,K,r_+}^{n',n',R_+}E_{K,K,r_-}^{n',n',R_-}.
\end{split}
\end{equation}
\end{proof}

Let $\mathcal{Z}_{R_+,R_-}(U_1,U_2)$ denote the set of reduced cycles for $\partial=\delta_1\otimes I_2+P_1\otimes\delta_2$ with block ranks $R_\pm$, where $U_i$ are block diagonal with two $n\times n$ blocks and
\begin{equation}
\delta_i=U_i
\begin{bmatrix}
0 & \delta_0  \\
\delta_0 & 0
\end{bmatrix}U_i^{-1}.
\end{equation}
By definition, $|\mathcal{Z}_{R_+,R_-}(U_1,U_2)|=\Gamma(R_+,R_-)$ when the $\delta_i$ are good.
\begin{lem}\label{lem:goodparam}
It is possible to parameterize the sets $\mathcal{Z}_{R_+,R_-}(U_1,U_2)$ with integers $j=1,\ldots,\Gamma(R_+,R_-)$ in such a way that conditioned on $\delta_1,\delta_2$ being good, the distribution of the $j$th reduced cycle is uniform on the set of pairs of $n'\times n'$ matrices of ranks $(R_+,R_-)$.
\end{lem}
\begin{proof}
Let $\delta_i=V_i
\begin{bmatrix}
0 & \delta_0  \\
\delta_0 & 0
\end{bmatrix}V_i^{-1}$ be good boundary operators. Consider the subgroup $\GL(n',\mathbb{Z}_D)^2\le\GL(n,\mathbb{Z}_D)^2$ of matrices
\begin{equation}
U=
\begin{bmatrix}
U'_+ & 0 & 0 & 0  \\
0 & I & 0 & 0  \\
0 & 0 & U'_- & 0  \\
0 & 0 & 0 & I
\end{bmatrix},
\end{equation}
where the block sizes are $n',n-n',n',n-n'$. If $U_i$ are such matrices, let $U'_i=U'_{i+}\oplus U'_{i-}$. Since $\delta_i$ is good and conjugation by $U_i$ does not affect the last $n-n'$ coordinates, $U_i\delta_i U_i$ is also good and $\mathcal{Z}_{R_+,R_-}(U_1V_1,U_2V_2)=(U'_1\otimes U'_2)\mathcal{Z}_{R_+,R_-}(V_1,V_2)$.

Choose an arbitrary parameterization $f_{V_1,V_2}:[\Gamma(R_+,R_-)]\to\mathcal{Z}_{R_+,R_-}(V_1,V_2)$ and let $f_{U_1V_1,U_2V_2}(j)=(U'_1\otimes U'_2)f_{V_1,V_2}(j)$ for all $U_1,U_2$ as above. Doing this for one representative $(V_1,V_2)$ in each coset gives a parameterization for all good pairs which has the desired properties.
\end{proof}

\begin{lem}\label{lem:uniformlowprob}
For any $\epsilon>0$ there is a $c'>0$ such that for any $1\le R\le n'$ the probability that a random $n'\times n'$ matrix $Z$ chosen uniformly from the rank $R$ matrices has row and column weights at most $c'$ is $O(1)\cdot D^{R^2-2(1-\epsilon)n'R}$.
\end{lem}
\begin{proof}
Similarly as \cite[Lemma 14]{BravyiHastings} (Lemma 6 in \cite{BravyiHastingsACM}).
\end{proof}

\begin{lem}\label{lem:Predgood}
Let $\delta_1,\delta_2$ be random boundary operators. Let $P_{\text{red}}^{\text{good}}$ denote the probability that there is a nonzero reduced cycle for $\partial=\delta_1\otimes I_2+P_1\otimes\delta_2$ satisfying the uniform low weight condition conditioned on $\delta_1$ and $\delta_2$ being good. Then for any $\epsilon>0$ and $0<r<1$ there is a $c>0$ such that
\begin{equation}
P_{\text{red}}^{\text{good}}\le O(n^6)\cdot D^{-2(1-2r-2\epsilon+2r\epsilon-\rho)n}.
\end{equation}
\end{lem}
\begin{proof}
Choose a parameterization of the reduced cycles as in lemma \ref{lem:goodparam}. Then the $j$th reduced cycle in $\mathcal{Z}_{R_+,R_-}(U_1,U_2)$ is distributed uniformly on the set of pairs of $n'\times n'$ matrices with ranks $R_\pm$. If $c$ is small enough, then lemma \ref{lem:uniformlowprob} implies that the probability that the $j$th reduced cycle satisfies the uniform low weight condition (with $c'=cr^{-1}/(1-r)$) is $O(1)\cdot D^{R_+^2-2(1-\epsilon)n'R_+}D^{R_-^2-2(1-\epsilon)n'R_-}$. Summing over every reduced cycle gives
\begin{equation}
P_{\text{red}}^{\text{good}}\le O(1)\sum_{R_\pm=1}^{n'}\Gamma(R_+,R_-)D^{R_+^2+R_-^2-2(1-\epsilon)n'(R_++R_-)}.
\end{equation}

Next we plug in the bounds for $\Gamma,Z$ and $E$ from lemmas \ref{lem:gammasum}, \ref{lem:Zbound} and \ref{lem:Ebound}. Then we extend the sums over $r_\pm$ to $R_\pm$ and exchange the order of the sums to get
\begin{multline}
P_{\text{red}}^{\text{good}}\le O(n)\sum_{l=0}^{2n'-n-H}D^{-l^2-2Hl}  \\
\cdot\left(\sum_{r_+=0}^{n'}\sum_{R_+=\max\{1,r_+\}}^{n'}D^{R_+^2-2(1-\epsilon)n'R_+}D^{nR_+-\frac{3}{4}R_+^2}D^{(H+R_+/2)r_+-\frac{3}{4}r_+^2+lr_+}\right)^2,
\end{multline}
where the square appears because the sum over $(r_-,R_-)$ gives the same factor.

The exponent in the inner sum is quadratic, for a fixed $r_+$ it has a minimum at $R_+=2n-r_+-4n(r+\epsilon-r\epsilon)$, which is larger than the upper endpoint, therefore the maximum is at $R_+=\max\{1,r_+\}$. If $r_+=0$ then we get $\frac{1}{4}+n-2n'(1-\epsilon)$, while $r_+>0$ gives $(H+l+n-2n'(1-\epsilon))r_+$. If $l>-H-n+2n'(1-\epsilon)$ this is increasing, therefore the maximum is at $r_+=n'$, otherwise it is decreasing and the maximum is at $r_+=1$. We get an upper bound by keeping the three possible maxima and counting the number of terms ($O(n^2)$). The square of a sum is less than the sum of squares times the number of terms, which is bounded, therefore
\begin{multline}
P_{\text{red}}^{\text{good}}\le O(n^5)\sum_{l=0}^{2n'-n-H}D^{-l^2-2Hl}\Big(D^{\frac{1}{2}+2n-4n'(1-\epsilon)}+D^{2(H+l+n-2n'(1-\epsilon))n'}  \\  +D^{2(H+l+n-2n'(1-\epsilon))}\Big).
\end{multline}

After expanding the product, each exponent is quadratic in $l$ and they have maxima at $0$, $0$ and $n'-H$, respectively. Multiplying the maximum with the number of terms in the sum gives
\begin{multline}
P_{\text{red}}^{\text{good}}\le O(n^6)\Big(D^{\frac{1}{2}-2n(1-2r-2\epsilon+2r\epsilon)}+D^{-n^2(1-4r+3r^2-4\epsilon+8r\epsilon-4r^2\epsilon-\rho^2)}  \\  +D^{-2n(1-2r-2\epsilon+2r\epsilon-\rho)}\Big),
\end{multline}
where we used $H=\rho n$ and $n'=(1-r)n$. When $n$ is large, the last term dominates.
\end{proof}

We are in the position now to state the main theorem. Recall that in coding
theory, the notation $[n,k,d,w]$ refers to a code of size $n$, encoding $k$
logical qudits, with distance $d$, and stabilizer weight $w$.

\begin{thm}
For any large enough $n\in\mathbb{N}$ there exist CSS codes with parameters $[2n^2,2\rho^2 n^2,c n^2,2n]$.
\end{thm}
\begin{proof}
Let $\delta_1,\delta_2$ be random boundary operators and $\partial=\delta_1\otimes I_2+P_1\otimes \delta_2$ as before. Then $CSS(C\otimes C,\partial,P_1\otimes P_2)$ encodes $2(\rho n)^2$ qudits and its stabilizer weights are at most $2n$.

The probability that there is a choice of $n'$ rows and columns in both blocks such that there is a nonzero reduced cycle (on these rows and columns) obeying the uniform low weight condition can be bounded from above as
\begin{equation}
P^{\text{bad}}+(1-P^{\text{bad}})\binom{n}{n'}^4P_{\text{red}}^{\text{good}}\le o(1)+O(n^6)\binom{n}{n'}^4D^{-2(1-2r-2\epsilon+2r\epsilon-\rho)n}
\end{equation}
by lemmas \ref{lem:randomcode} and \ref{lem:Predgood}. Using that
\begin{equation}
\binom{n}{n'}=\binom{n}{(1-r)n}=D^{(-r\log_Dr-(1-r)\log_D(1-r))n+o(n)}
\end{equation}
we get that for sufficiently small $r$ and $\epsilon$ this probability goes to $0$. The same bound applies to cocycles since $\partial$ and $\partial^T$ have the same distribution. Therefore when $n$ is large enough, there is at least one $\delta_1,\delta_2$ such that $CSS(C\otimes C,\partial,P_1\otimes P_2)$ has distance $\ge cn^2$.
\end{proof}

\end{document}